\documentclass[conference]{IEEEtran}
\IEEEoverridecommandlockouts

\usepackage{cite}
\usepackage{amsmath,amssymb,amsfonts}
\usepackage{algorithmic}
\usepackage{graphicx}
\usepackage{textcomp}
\usepackage{xcolor}
\usepackage{amsmath, amssymb, amsthm}
\theoremstyle{definition}
\usepackage{mathtools}
\newtheorem{definition}{Definition}
\newtheorem{example}{Example}
\theoremstyle{remark} 

\newtheorem{theorem}{Theorem}
\newtheorem{lemma}{Lemma}
\newtheorem{corollary}{Corollary}
\usepackage{amsmath}
\usepackage{tabularx}

\usepackage{tikz} 
\usetikzlibrary{shapes}
\usepackage{pgfplots}
\pgfplotsset{compat=1.5}
\usepackage[utf8]{inputenc} 
\usepackage[T1]{fontenc}
\usepackage{url}
\usepackage{ifthen}
\usepackage{graphicx}
\usepackage{makecell}
\usepackage{textcomp}
\usepackage{xcolor}
\usepackage{caption}
\usepackage{mathtools}
\usepackage{comment}
\usepackage{amsthm}
\usepackage{csquotes}
\usepackage{enumitem}
\usepackage{xcolor, soul}
\sethlcolor{yellow}
\usepackage{arydshln}
\usepackage{stfloats}
\usepackage{epsfig}

\def\BibTeX{{\rm B\kern-.05em{\sc i\kern-.025em b}\kern-.08em
    T\kern-.1667em\lower.7ex\hbox{E}\kern-.125emX}}
\usepackage{float} 
\begin{document}

\title{Plotkin-type Bound on Function-Correcting Codes for Monomial Mappings}
\author{
\IEEEauthorblockN{Kanchana Lokshmii Jagatti and B. Sundar Rajan}
\IEEEauthorblockA{
Department of Electrical Communication Engineering\\
Indian Institute of Science, Bengaluru, India\\ 
Email: \{kanchanal, bsrajan\}@iisc.ac.in
}
}
\maketitle

\begin{abstract}
Function-correcting codes, recently introduced by Lenz et al.\ (2023), provide a framework for the receiver to provide a higher level of protection against channel errors than the level of protection for messages.  In this paper, we derive a generalised Plotkin-type bound on the redundancy of FCCs under the Hamming metric for arbitrary functions over finite fields. The proposed bound recovers the classical Plotkin bound for classical error-correcting codes as a special case. We further obtain explicit Plotkin-type bounds for linear functions, Hamming-weight functions, and Hamming-weight threshold functions. In addition, we study FCCs for monomial mappings $f(x)=x^n$ over $\mathbb{F}_{q^k}$, which are widely utilised in coding theory, cryptography, and finite geometry, and derive an explicit Plotkin-type redundancy bound. 
\end{abstract}

\section{Introduction}
\label{sec1}
In several modern communication and distributed computing systems, the decoding objective is not necessarily only the full reconstruction of the transmitted message, but also the reliable computation of a prescribed function of the message  Function-correcting codes (FCC), recently introduced by Lenz et al.\ (2023), provide a framework for receiver to provide higher level of protection against channel errors than the level of protection for messages  When the function is a bijective map the theory of FCC coincides with the theory of classical Error Corecting Codes ECC)  FCCs ensure that a function \( f(u) \) can be computed correctly even when the encoded representation of the message \( u \) is corrupted by up to \( t \) errors  Lower bounds on the redundancy of FCCs, for several functions, have been reported in \cite{1} and for linear functions in \ cite {2}. 

The study of FCCs has since evolved along multiple directions  The fundamental limits on redundancy for various classes of functions have been characterised in~\cite{3,4}, while extensions to different read channels, including symbol-pair and \( b \)-symbol channels, have been investigated in~\cite{5}--\cite{6}  More recently, structural properties such as homogeneous distance constraints and locally bounded function classes have been analysed in~\cite{7}  Additionally, refined constructions such as partition-based FCCs and FCCs with explicit data protection guarantees have been proposed in~\cite{8,9}, demonstrating that FCCs with identical function-correcting capability may exhibit significantly different performance with respect to data recovery.

A common feature in the existing literature is that the function \( f \) is either treated as an arbitrary mapping or is constrained via local properties  However, functions induced by algebraic structures, particularly monomial mappings over finite fields, have not been systematically studied within the FCC framework  Monomial mappings induce structured many-to-one mappings, thereby partitioning the input space $( \mathbb{F}_q^k )$ into equivalence classes defined by the preimages of function values  Such mappings exhibit rich combinatorial and algebraic properties that are highly relevant for coding-theoretic analysis  Recent studies on many-to-one mappings over finite fields~\cite{11}, two-to-one mappings~\cite{12}, and characterisations and constructions of (n)-to-(1) mappings~\cite{13} further highlight the significance of algebraically structured mappings in finite field theory and coding theory, motivating their investigation in the context of function-correcting codes.

In this paper, we first obtain a Plotkin-type lower bound on the redundancy of FCCs for any arbitrary function defined on the message space $\mathbb{F}_q^k$, which is computationally more efficient than the similar bounds available in the literature for arbitrary functions. Then we obtain a closed-form, simplified Plotkin-type bound for the Hamming weight function and the monomial mapping functions.  

The main contributions of this work are summarised as follows:

\begin{itemize}

\item We derive a generalised Plotkin-type redundancy bound for function-correcting codes under the Hamming metric and show that it recovers the classical Plotkin bound for error-correcting codes as a special case.

\item We derive simplified Plotkin-type bounds for linear functions, Hamming-weight functions, and Hamming-weight threshold functions.

\item We derive an explicit Plotkin-type redundancy bound for monomial mapping functions over finite fields using their preimage-set structure.

\end{itemize}

\subsection{Notations}
 The set of all positive integers is denoted by $\mathbb{N}$  The finite field with $q^k$ elements is denoted by $\mathbb{F}_{q^k}$, where $q$ is a prime power  The multiplicative group of nonzero elements of $\mathbb{F}_{q^k}$ is denoted by $\mathbb{F}_{q^k}^{\ast}$. The Hamming weight of a vector $\mathbf{x}\in\mathbb{F}_q^k$ is denoted by $w_H(\mathbf{x})$, and the Hamming distance between vectors $\mathbf{x},\mathbf{y}\in\mathbb{F}_q^k$ is denoted by $d_H(\mathbf{x},\mathbf{y})$. For a vector $\mathbf{x}$ of length $k$, $x_i$ denotes its $i$th coordinate. The notation $\mathbf{1}(\mathcal{E})$ equals $1$ if the event $\mathcal{E}$ holds and $0$ otherwise. 
\section{Preliminaries}
\label{sec2}
The system model consists of a transmitter that wants to send a message vector $x \in \mathbb{F}_q^k$  Since $\mathbb{F}_{q^k}$ can be viewed as a $k$-dimensional vector space over $\mathbb{F}_q$, each message vector in $\mathbb{F}_q^k$ may equivalently be represented as an element of $\mathbb{F}_{q^k}$. The receiver is interested in evaluating a function
$f : \mathbb{F}_{q^k} \to \mathrm{Im}(f)$,
where $q=p^s$, with $p$ prime and $s$ denoting the extension degree. The transmitter employs a systematic encoding
$\mathcal{E} : \mathbb{F}_q^k \to \mathbb{F}_q^{k+r}$
such that $\mathcal{E}(x)=(x,p)$. The encoded vector $\mathcal{E}(x)$ is transmitted over a noisy channel that may introduce up to $t$ number of $\mathbb{F}_q$-symbol errors. The receiver obtains $y=\mathcal{E}(x)+e \in \mathbb{F}_q^{k+r}$,
where $w_H(e)\leq t$, and applies a decoding function
$\mathcal{D}:\mathbb{F}_q^{k+r}\to \mathrm{Im}(f)$ such that
$\mathcal{D}(y)=f(x)$ for all $x\in\mathbb{F}_q^k$.

\subsection{Function-Correcting Codes}
\begin{definition}~\cite{1}
  A systematic encoding $\mathcal{E} : \mathbb{F}_q^k \rightarrow \mathbb{F}_q^{k+r}$ is said to be an $(f,t)$-FCC for a function $f : \mathbb{F}_q^k \rightarrow Im(f)$ if for all $x_i, x_j \in \mathbb{F}_q^k$ with $f(x_i) \neq f(x_j)$, we have
\begin{equation*}
d_H(\mathcal{E}(x_i),\mathcal{E}(x_j)) \geq 2t + 1.
\end{equation*}
\end{definition}

\noindent \textbf{Note:} For any bijective function $f$, an $(f,t)$-FCC is equivalent to a standard $(n, q^k, 2t+1)$ error-correcting code.

\begin{definition}[\cite{1}]
The optimal redundancy $r_f(k,t)$ is defined as the smallest integer $r$ such that there exists an $(f,t)$-FCC with encoding function $\mathcal{E}: \mathbb{F}_q^k \rightarrow \mathbb{F}_q^{k+r}$.
\end{definition}

\begin{definition}[\cite{1}]
Let $x_0, x_1, \dots, x_{q^k-1} \in \mathbb{F}_q^k$. The distance requirement matrix (DRM) $D = D_f(t, x_0, \dots, x_{q^k-1})$ for an $(f,t)$-FCC is a $q^k \times q^k$ matrix with entries
\[
D_{ij} =
\begin{cases}
 \max(0, 2t + 1 - d_H(x_i, x_j) ), & \text{if } f(x_i) \neq f(x_j), \\
 0, & \text{otherwise}.
\end{cases}
\]
\end{definition}

The DRM depends on the ordering of the elements of $\mathbb{F}_q^k$.

\begin{definition}\cite{1}
Let $D \in \mathbb{N}_0^{M \times M}$. A set $\mathcal{P} = \{p_i : i \in [M]\} \subset \mathbb{F}^r$ is said to be an irregular-distance code with constraint $D$ (or a $D$-code) if there exists an ordering of $\mathcal{P}$ such that
\[
d_H(p_i, p_j) \geq D_{ij}, \quad \forall i,j \in [M].
\]
Further, $N(D)$ is defined as the smallest integer $r$ such that there exists a $D$-code of length $r$ over $\mathbb{F}_q$.
\end{definition}

\noindent For the case where $D = D_f(t, x_0, \dots, x_{q^k-1})$, we take $\mathcal{P} = \{p_i : i = 0,1,\dots,q^k-1\}$, and the encoding of an $(f,t)$-FCC can be constructed as  $\mathcal{E}(x_i) = (x_i, p_i).$
\subsection{Many-to-One Monomial Mapping}

Many-to-one mappings have wide applications in several areas, particularly in cryptography, finite geometry, coding theory, and combinatorial design theory~\cite{11,12,13}. The following definitions are from~\cite{11}. 

\begin{definition}\cite{11}
Let $A$ be a finite set and  $m \in \mathbb{Z}$ with
$1 \leq m \leq |A|$. Write $|A| = km + r,$ where $k,r \in \mathbb{Z}$ and $0 \leq r < m$. Let $f: A \to B$ be a mapping from $A$ to another finite set $B$. Then $f$ is called an $m$-to-$1$ mapping on $A$ if there exist $k$ distinct elements in $B$ such that each of these elements has exactly $m$ preimages in $A$ under $f$.
The remaining $r$ elements in $A$ are called the exceptional elements of $f$ on $A$, and the set of these $r$ exceptional elements is called the exceptional set of $f$ on $A$, denoted by
$E_f(A)$.
In particular, $E_f(A)=\emptyset$ if and only if $r=0,$
that is, $m \mid |A|.$
\end{definition}

\begin{definition}\cite{11}
A polynomial $f(x) \in \mathbb{F}_q[x]$ is called \emph{$m$-to-1}, over $\mathbb{F}_q$ if the induced mapping 
$f : \mathbb{F}_q \to \mathbb{F}_q$, defined by $c \mapsto f(c)$, is $m$-to-1 on $\mathbb{F}_q$.
\end{definition}

\begin{example}
 Let $f(x) = x^3 + x$. Then $f$ maps $0,1,2,3,4$ to $0,2,0,0,3$ in $\mathbb{F}_5$, respectively. Thus, $f$ is $3$-to-$1$ on $\mathbb{F}_5$, and the exceptional set $E_f(\mathbb{F}_5) = \{1,4\}$.   
\end{example}

\begin{theorem}  \cite{11}
	\label{thm1}
The monomial $f(x)=x^n$ with $n \in \mathbb{N}$ is $gcd(n, q-1)$-to-$1$ on $\mathbb{F}_q^{*}$, and $E_f(\mathbb{F}_q^{*}) = \emptyset$.
\end{theorem}
\section{Plotkin-type Bound for Arbitrary Function}
\label{sec3}
In this section we obtain the main result which a Plotkin-type bound for FCCs with Hamming metric for arbitrary functions. 

Let $f: \mathbb{F}_{q^k} \rightarrow \operatorname{Im}(f)$ be a mapping whose image contains $m$ distinct elements, i.e., $|\operatorname{Im}(f)|=m$. Let the distinct preimage sets of $f$ have cardinalities $\ell_1,\ell_2,\ldots,\ell_m$. Let 
 $\mathcal A_g,$  $g=1,2,\cdots,m,$ denote the $g$-th distinct preimage set of cardinality $\ell_g$. 
 \begin{definition}
 	 Let the function $f(x)$ partition the domain $\mathbb{F}_{q^k}$ into $m$ equivalence classes denoted by 
 	$\mathcal A_g$ where each equivalence class is a preimage set  We call  $S_g=\sum_{x_i,x_j\in\mathcal A_g}d_H(x_i,x_j)$ the intra-preimage distance for the $g$-th partition class and  $S=\sum_{g=1}^{m} S_g$ the total intra-preimage distance.
 \end{definition}
\begin{theorem}
	\label{thm2}
	Let $f: \mathbb{F}_{q^k} \rightarrow \operatorname{Im}(f)$ be a mapping whose image contains $m$ distinct elements, i.e., $|\operatorname{Im}(f)|=m$. Let the distinct preimage sets of $f$ have cardinalities $\ell_1,\ell_2,\ldots,\ell_m$. Let $D$ be the distance requirement matrix (DRM) corresponding to an error-correction capability $t$. Then, the redundancy $r_f(k,t)$ satisfies
	\begin{equation*}
		r_f(k,t)\ge\frac{(2t+1)\left(q^{2k}-\sum_{g=1}^{m}\ell_g^2\right)+S}  {q^{2k-1}(q-1)}-k,
	\end{equation*}
	where $S$ is the total intra-preimage distance. 
\end{theorem}

\begin{proof}
	The bound is obtained by counting the sum $\sum_{i,j} d_H(p_i, p_j)$ in two different ways  Let $D$ denote the distance requirement matrix (DRM)  Let $p_1, \ldots, p_{q^k}$ be the codewords of a $D$-code of length $r$, arranged as a $q^k \times r$ matrix $P$. 
	
	Since each column of $P$ contributes at most $q^{2k}\left(1 - \frac{1}{q}\right) = q^{2k-1}(q-1)$ to the sum, we have 
	\begin{equation} 
		\sum_{i,j} d_H(p_i, p_j) \le r\, q^{2k-1}(q-1). 
		\label{eq:upper_bound} 
	\end{equation} 
	
	By definition, $d_H(p_i, p_j) \ge D_{i,j}$, which implies 
	\begin{equation}
		\sum_{i,j} d_H(p_i, p_j) \ge \sum_{i,j} D_{i,j}.
		\label{eq:sum_d}
	\end{equation}
	
	By combining \eqref{eq:sum_d} and \eqref{eq:upper_bound}, we obtain
	\begin{equation} 
		r \ge \frac{1}{q^{2k-1}(q-1)} \sum_{i,j} D_{i,j}. 
		\label{eq:r_sum_main} 
	\end{equation} 
  Let the $m$ distinct preimage sets be denoted by $\mathcal{A}_1, \ldots, \mathcal{A}_m$, where $\ell_g = |\mathcal{A}_g|$ represents the size of the $g$-th preimage set. For any element $x_j \in \mathcal{A}_g$, its corresponding preimage set is $\mathcal{A}_g = \{x_i \in \mathbb{F}_q^k : f(x_i) = f(x_j)\}$.
    By partitioning the rows and columns of the matrix according to these disjoint preimage sets, the sum of all entries of the DRM can be written as
    \begin{equation}
        \sum_{i,j} D_{i,j} = \sum_{g=1}^{m} \sum_{x_j \in \mathcal{A}_g} \sum_{i} D_{i,j}.
        \label{eq:partition}
        \end{equation}
        For a fixed element $x_j \in \mathcal{A}_g$, the corresponding column sum is given by
        \begin{align} 
        &\sum_i D_{i,j} 
        = (2t+1)(q^k-\ell_g) - \sum_{f(x_i) \neq f(x_j)} d_H(x_i,x_j) \notag \\ 
        &= (2t+1)(q^k-\ell_g) - \sum_i d_H(x_i,x_j) + \sum_{x_i \in \mathcal{A}_g} d_H(x_i,x_j). 
        \label{eq:nonzero_col_step} 
        \end{align}
        Since $\sum_i d_H(x_i,x_j)=k(q-1)q^{k-1}$ for every $x_j\in\mathbb{F}_{q}^k$, substituting this identity into \eqref{eq:nonzero_col_step} yields
        \begin{align}
            \sum_i D_{i,j} &= (2t+1)(q^k-\ell_g) - k(q-1)q^{k-1} \notag \\ 
              &\quad + \sum_{x_i \in \mathcal{A}_g} d_H(x_i,x_j).
        \label{eq:column_sum_final}
        \end{align}
        Substituting \eqref{eq:column_sum_final} into \eqref{eq:partition}, we get
        \begin{align*}
        \sum_{i,j} D_{i,j}&= \sum_{g=1}^{m} \sum_{x_j \in \mathcal{A}_g} \Bigl( (2t+1)(q^k-\ell_g) - k(q-1)q^{k-1} \Bigr) \notag \\&\quad + \sum_{g=1}^{m} \sum_{x_j \in \mathcal{A}_g} \sum_{x_i \in \mathcal{A}_g} d_H(x_i,x_j)\notag \\
        \end{align*}
        \begin{align}
            = (2t+1)\sum_{g=1}^{m} \ell_g(q^k-\ell_g) - k(q-1)q^{k-1}\sum_{g=1}^{m}\ell_g + S\label{eq:sumD}
        \end{align}
        Using the identity $\sum_{g=1}^{m}\ell_g=q^k$, we observe that 
        \begin{equation}
        \sum_{g=1}^{m}\ell_g(q^k-\ell_g) = q^k\sum_{g=1}^{m}\ell_g - \sum_{g=1}^{m}\ell_g^2 = q^{2k} - \sum_{g=1}^{m}\ell_g^2.
        \end{equation}
        Consequently, \eqref{eq:sumD} reduces to
        \begin{equation}
        \sum_{i,j}D_{i,j} = (2t+1) \left( q^{2k} - \sum_{g=1}^{m}\ell_g^2 \right) - k(q-1)q^{2k-1} + S.
        \label{total_d}
        \end{equation}
        Substituting \eqref{total_d} into \eqref{eq:r_sum_main} yields the final bound
        \begin{align}
        r_f(k,t) &\ge \frac{(2t+1)\left(q^{2k} - \sum_{g=1}^{m}\ell_g^2 \right) - k(q-1)q^{2k-1} + S}{q^{2k-1}(q-1)} \notag \\
        &= \frac{(2t+1)\left(q^{2k} - \sum_{g=1}^{m}\ell_g^2 \right) + S}{q^{2k-1}(q-1)} - k,
        \label{eqnmain}
        \end{align}
        which completes the proof.
    \end{proof}
\subsection{Recovering Plotkin Bound for classical ECCs} 

In this section, we recover the classical ECC Plotkin bound for the Hamming metric by considering the function to be a bijective map with $l_g = 1$  The bound stated in ~\eqref{eqnmain} simplifies to 
\begin{equation} 
	r \ge \frac{(2t+1)(q^{2k} - q^k)}{q^{2k-1}(q-1)} - k. 
	\label{r_plotkin}
\end{equation} 
Substituting $d = 2t+1$ and $M = q^k$ in ~\eqref{r_plotkin}, we obtain 
\begin{align*} 
	r &\ge \frac{q^k(q^k - 1)}{q^{2k-1}(q-1)} d - k 
	= \frac{q}{q-1} d \left(1 - \frac{1}{M}\right) - k. 
\end{align*} 
Rearranging the terms yields
\begin{equation*}
	\frac{1}{M} \ge 1 - \frac{q-1}{q}\cdot\frac{r+k}{d}
	= \frac{qd - (q-1)(r+k)}{qd}.
\end{equation*}
Hence,
\begin{equation*}
	M \le \frac{q d}{q d - (q-1)(r+k)}.
\end{equation*}
This corresponds to the classical Plotkin bound for error-correcting codes (ECCs) with block length $n =k+ r$ and minimum distance $d$, as stated in  \cite[Theorem 2.2.1]{10}.

\begin{example}
	We consider a bijective mapping $\ell=1$ over $\mathbb{F}_{16}$ ($k=4, q=2$) with $t=1$. Substituting these parameters  in~\eqref{r_plotkin} we get $r \ge 2$, which aligns with the parameters of the classical $(7, 4, 3)$ Hamming code. In this case, $r=3$ parity bits are required to protect $k=4$ information bits against $t=1$ error.
\end{example}

Notice that the bound in Theorem \ref{thm2} is valid for any arbitrary function on $\mathbb{F}_{q^k}.$ It depends only on the size of the pre-image sets and the intra-preimage distance of the pre-image sets. This means that the bound is the same for all the functions that create the same partition on the domain $\mathbb{F}_{q^k}.$ 
\subsection{Plotkin-type Bound for linear functions}
In this subsection, we consider linear functions on $\mathbb{F}_{q}^k$ and present a simplified Plotkin-type bound that depends on a single parameter.

If $f:\mathbb{F}_q^k \rightarrow \mathbb{F}_q^l$ is a linear function, then $\operatorname{Im}(f)=\mathbb{F}_q^l$ and hence $m=q^l$. Moreover, $\dim(\ker(f))=k-l$, so every preimage set is a coset of $\ker(f)$ and has cardinality $\ell_g=q^{k-l}$ for $g=1,\ldots,m.$ Consequently,$\sum_{g=1}^{m}\ell_g^2 = q^l(q^{k-l})^2 = q^{2k-l}$. Since every preimage set is a coset of $\ker(f)$, any pair of elements $x_i,x_j\in\mathcal A_g$ satisfies $x_i=x_j+u$ for a unique $u\in\ker(f)$. 
Hence,  $d_H(x_i,x_j) = d_H(x_j+u,x_j) = d_H(u,0) = w_H(u)$. Substituting this relation into the definition of $S$ yields 
\begin{equation*}
	S = \sum_{g=1}^{m} \sum_{x_j\in\mathcal A_g} \sum_{u\in\ker(f)} w_H(u) = \sum_{g=1}^{m} |\mathcal A_g| \sum_{u\in\ker(f)} w_H(u). 
\end{equation*}
Given that $|\mathcal A_g|=q^{k-l}$ for every $g$ and the number of preimage sets is $m=q^l$, we obtain the collapsed distance expression
\begin{equation}
	S = q^l\cdot q^{k-l} \sum_{u\in\ker(f)} w_H(u) = q^k \sum_{u\in\ker(f)} w_H(u).
	\label{s_linear}
\end{equation}

Therefore, substituting \eqref{s_linear} and $\sum_{g=1}^{m}\ell_g^2$ into Theorem~2 reduces the lower bound to
\begin{equation*}
	r_f(k,t) \ge \frac{q}{q-1}(2t+1)(1-q^{-l}) + \frac{\textstyle\sum_{u\in\ker(f)} w_H(u)}{(q-1)q^{k-1}} - k,
\end{equation*}
\begin{equation*}
	r_f(k,t) \ge \frac{q}{q-1}(2t+1)(1-q^{-l}) + \frac{s}{(q-1)q^{k-1}} - k, 
\end{equation*}
where $s$ is the sum of the Hamming weights of the vectors in $ker(f).$ This is precisely the Plotkin-type bound for linear functions presented in~\cite{2}.
\subsection{A Plotkin-type Bound for the Hamming Weight Function}

We now apply Theorem \ref{thm2} to the Hamming weight function $f(x)=w_H(x)$ for $x\in\mathbb F_q^k$. 
The image of the Hamming weight function consists of all possible weights from $0$ to $k$ is $\operatorname{Im}(f) = \{0, 1, \dots, k\}, \quad \text{implying} \quad |\operatorname{Im}(f)|  = k+1.$ For each weight $w \in \{0, 1, \dots, k\}$, the corresponding preimage set $\mathcal A_w$ contains all vectors in $\mathbb F_q^k$ with exactly $w$ nonzero coordinates. 
Thus, the cardinality of each preimage set is
$l_w = |\mathcal A_w| = \binom{k}{w}(q-1)^w.$

To compute the bound in Theorem \ref{thm2}, we calculate the sum of the squares of these preimage sizes
\begin{equation}\label{eq:sum_squares_q}
	\sum_{g=1}^{m}\ell_g^2 = \sum_{w=0}^{k} \ell_w^2 = \sum_{w=0}^{k} \binom{k}{w}^{2}(q-1)^{2w}.
\end{equation}

Substituting \eqref{eq:sum_squares_q} directly into the general redundancy bound of Theorem \ref{thm2} yields
\begin{equation}\label{eq:general_bound}
	r_f(k,t) \ge \frac{ (2t+1) \left( q^{2k} - \sum_{w=0}^{k} l_w^2 \right) + S }{q^{2k-1}(q-1)} - k,
\end{equation}

where $S$ represents the total intra-preimage distance. Now we proceed to show that $S = \sum_{w=0}^{k} |A_w|^2 \left( 2w-\frac{qw^2}{k(q-1)} \right)$. We have $S_w=\sum_{x,y\in\mathcal A_w} d_H(x,y)$, and utilizing the coordinate identity $d_H(x,y)=\sum_{i=1}^{k}\mathbf 1(x_i\neq y_i)$, we obtain $S_w=\sum_{i=1}^{k}\sum_{x,y\in\mathcal A_w}\mathbf 1(x_i\neq y_i)$.

	Let $N_{w,i} = \sum_{x,y\in\mathcal A_w} \mathbf 1(x_i\neq y_i)$ for $1\le i\le k$, which counts the number of ordered pairs in $\mathcal |A_w|^2$ that differ in the $i$-th coordinate. Since $\mathcal A_w$ is invariant under coordinate permutations, the quantities $N_{w,i}$ are the same for all $i$, allowing us to denote this common value simply as $N_w$. Therefore, the total distance sum evaluates to $S_w = \sum_{x,y\in\mathcal A_w} d_H(x,y) = \sum_{i=1}^{k} N_{w,i} = kN_w$.
	
	To determine $N_w$, let $B_a=\bigl|\{x\in\mathcal A_w:x_i=a\}\bigr|$ for $a\in\mathbb F_q$ and any fixed coordinate $i$. Then the number of ordered pairs that agree in this coordinate is $\sum_{a\in\mathbb F_q} B_a^2$. Since there are $\ell_w^2$ ordered pairs in total, it follows that $N_w=\ell_w^2-\sum_{a\in\mathbb F_q} B_a^2$.
	
	We evaluate these cardinalities by splitting $a$ into the zero and nonzero cases:
	\begin{align*}
		B_0 &= \binom{k-1}{w}(q-1)^w, \quad \text{for } a = 0, \\
		B_a &= \binom{k-1}{w-1}(q-1)^{w-1}, \quad \text{for } a \in \mathbb{F}_q^*.
	\end{align*}
	Since there are exactly $q-1$ elements where $a \in \mathbb{F}_q^*$, separating the $a=0$ term from the rest of the summation yields
	\[
	\sum_{a\in\mathbb F_q} B_a^2 = \binom{k-1}{w}^{2}(q-1)^{2w} + (q-1)\binom{k-1}{w-1}^{2}(q-1)^{2w-2}.
	\]
	Substituting this summation into the expression $N_w = \ell_w^2 - \sum_{a\in\mathbb F_q} B_a^2$ and using $S_w = k N_w$, the total distance sum for weight class $w$ expands and simplifies as follows
	\begin{align*}
		S_w
		&= k\left[ \ell_w^2 - B_0^2 - (q-1)B_a^2 \right] \\
		&= k\left[ \ell_w^2 -\frac{(k-w)^2}{k^2}\ell_w^2 -\frac{w^2}{k^2(q-1)}\ell_w^2 \right] \\
		&= \ell_w^2 \left[ k-\frac{k^2-2kw+w^2}{k} -\frac{w^2}{k(q-1)} \right] \\
		&= \ell_w^2 \left[ 2w-\frac{q w^2}{k(q-1)} \right].
	\end{align*}
	
	Summing the contributions $S_w$ across all possible weight classes from $0$ to $k$ yields the final closed-form expression for the total intra-preimage distance sum
	\begin{align}\label{eq:s}
		S = \sum_{w=0}^{k} S_w 
		= \sum_{w=0}^{k} \ell_w^2 \left( 2w-\frac{q w^2}{k(q-1)} \right).
	\end{align}
    Substituting \eqref{eq:s} into \eqref{eq:general_bound} and simplifying, we obtain
    \begin{align}
    r_f(k,t)&\ge \frac{1}{q^{2k-1}(q-1)} \Big[ (2t+1)q^{2k} \nonumber \\
     +&\quad \sum_{w=0}^{k} l_w^2 \left( 2w-(2t+1) -\frac{q w^2}{k(q-1)} \right) \Big] - k.\label{eq:weight_bound_final}
    \end{align}
    For  $q=2$, substituting $l_w=\binom{k}{w}$  into \eqref{eq:weight_bound_final}, we obtain:
    \begin{align}
         r_f(k,t)
         &\ge 
        \frac{1}{2^{2k-1}} \Big[ (2t+1)4^k \nonumber \\
         +  &\quad\sum_{w=0}^{k} \binom{k}{w}^2 \left( 2w-\frac{2w^2}{k}-(2t+1) \right) \Big] - k.
    \end{align}
\begin{example}
Consider the Hamming weight function $f(x)=w_H(x)$ for $x\in\mathbb F_2^3$. The vectors are partitioned according to their weights. The corresponding pre-image cardinalities are given by $\ell_w=|\mathcal A_w|=\binom{3}{w}$, which yields the vector $(\ell_0,\ell_1,\ell_2,\ell_3)=(1,3,3,1)$. The total Hamming distance within these pre-image sets is defined as $S=\sum_{w=0}^{3}\sum_{x_i,x_j\in\mathcal A_w}d_H(x_i,x_j).$

Using the simplified $S=\sum_{w=0}^{k}\ell_w^2 \left(2w-\frac{qw^2}{k(q-1)}\right)$ with $q=2$ and $k=3$, yields a total sum of $S=24$.

Moreover, the sum of squares of the partition cardinalities is $\sum_{w=0}^{3}\ell_w^2 = 1^2+3^2+3^2+1^2 = 20.$
Substituting these parameters into the general redundancy bound expression for $t=1$ yields $r_f(3,1)\ge \frac{15}{8}$. Since the redundancy must be an integer, this bounds the metric to $r_f(3,1)\ge 2$.
\end{example}

\subsection{A Plotkin-type Bound for Hamming Weight  Threshold Function}
We now apply Theorem \ref{thm2} to the Hamming-weight threshold function $f(x) = \triangle_T(x) \triangleq \left\lfloor \frac{w_H(x)}{T} \right\rfloor, x \in \mathbb{F}_q^k, $where $T \in \mathbb{N}$ is chosen such that $T \mid (k+1)$. Let $E = \frac{k+1}{T}. $Then the image of $f$ consists of the $E$ distinct values $\operatorname{im}(f) = \{0, 1, \ldots, E-1\}, $  and hence $m = |\operatorname{im}(f)| = E$. For each $g \in \{0, 1, \ldots, E-1\}$, the corresponding preimage set is $\mathcal{A}_g = \left\{ x \in \mathbb{F}_q^k : gT \le w_H(x) \le (g+1)T-1 \right\}. $ Thus, $\mathcal{A}_g$ contains all vectors whose Hamming weights lie in the interval $[gT, (g+1)T-1]$. Consequently, the cardinality of $\mathcal{A}_g$ is 
\begin{equation*}
	\ell_g = |\mathcal{A}_g| = \sum_{w=gT}^{(g+1)T-1} \binom{k}{w}(q-1)^w. 
\end{equation*}
To apply Theorem \ref{thm2}, we compute the sum of the squared preimage sizes
\begin{equation*} 
	\sum_{g=1}^{m}\ell_g^2 = \sum_{g=0}^{E-1} \left( \sum_{w=gT}^{(g+1)T-1} \binom{k}{w}(q-1)^w \right)^2. 
\end{equation*} 
Substituting the above equation into Theorem \ref{thm2} yields the following lower bound on the redundancy, denoted simply as $r = r_f(k,t)$
\begin{equation} 
\resizebox{\columnwidth}{!}{$
r \ge
\frac{(2t+1)\Bigl(q^{2k}
-\sum_{g=0}^{E-1}
\bigl(\sum_{w=gT}^{(g+1)T-1}
\binom{k}{w}(q-1)^w\bigr)^2\Bigr)
+S}
{q^{2k-1}(q-1)}
-k
$}
\end{equation}
 
where  $S$ is the total intra-preimage distance. 

\subsection{Plotkin-Like Bound for Min-Max Function}

We now apply Theorem\ref{thm2} to the min-max function $\mathrm{mm}_w(u)$, where $u=(u^{(1)},\ldots,u^{(w)})\in\mathbb{F}_2^k$ with $k=w\ell$ and each block $u^{(i)}\in\mathbb{F}_2^\ell$. The min-max function is defined as $\mathrm{mm}_w(u)=\left(\arg\min_{1\le i\le w}u^{(i)}, \arg\max_{1\le i\le w}u^{(i)}\right)$, where the ordering among blocks is lexicographic.

The image of the min-max function consists of all ordered pairs $(i,j)$ such that $1\le i,j\le w$ and $i\neq j$. Therefore, $\operatorname{Im}(f)=\{(i,j):1\le i,j\le w,\ i\neq j\}$, which yields $|\operatorname{Im}(f)|=w(w-1)$. For each $(i,j)\in\operatorname{Im}(f)$, let $\mathcal A_{i,j} = \{u\in\mathbb F_2^k:\mathrm{mm}_w(u)=(i,j)\}$ denote the corresponding preimage set, and let $\ell_{i,j}=|\mathcal A_{i,j}|$.

To compute the bound in Theorem~\ref{thm2}, we require the sum of the squared preimage sizes $\sum_{g=1}^{m}\ell_g^2 = \sum_{\substack{1\le i,j\le w\\ i\neq j}} \ell_{i,j}^{\,2}.$

Substituting into Theorem~\ref{thm2} yields the redundancy bound
\begin{align*}
r_f(k,t) &\ge \frac{1}{2^{2k-1}} \Bigg[ (2t+1) \Bigg( 2^{2k}
 - \sum_{\substack{1\le i,j\le w\\ i\neq j}} \ell_{i,j}^{\,2} \Bigg) + S \Bigg] - k,
\label{eq:minmax_bound}
\end{align*}
where $S$  is the total intra-preimage distance sum, is defined as  
$S = \sum_{\substack{1\le i,j\le w\\ i\neq j}} \sum_{x,y\in\mathcal A_{i,j}} d_H(x,y).$

\section{FCC For Monomial  Mapping Functions}
In this section, we study monomial mappings, derive a Plotkin-type redundancy bound, and characterise the conditions under which these mappings are bijective. We consider monomial mappings of the form $f(x)=x^n$ over $\mathbb{F}_{q^k}$ and rewrite Theorem \ref{thm1} as the following lemma.

\begin{lemma}
Let $f : \mathbb{F}_{q^k} \to \mathrm{Im}(f)$ be the monomial mapping defined by $f(x)=x^n,$ where $n \in \mathbb{N}$, and let $\ell=\gcd(n,q^k-1).$ Then every nonzero element in $\mathrm{Im}(f)$ has exactly $\ell$ preimages in $\mathbb{F}_{q^k}^{\ast}$ and $|\mathrm{Im}(f)| = \frac{q^k-1}{\ell}+1.$
\end{lemma}

\begin{lemma}
Let $f:\mathbb{F}_{q^k}\to \mathrm{Im}(f)$ where $f(x)=x^n$, and define the kernel as $K=\ker(f)=\{x\in \mathbb{F}_{q^k}^{\ast}:x^n=1\}$. Then $K$ is a multiplicative subgroup of $\mathbb{F}_{q^k}^{\ast}$, and each nonempty preimage set of $f$ is a multiplicative coset of $K$.
\end{lemma}
\begin{proof}
Let
$K=\ker(f)=\{x\in\mathbb{F}_{q^k}^{\ast}:x^n=1\}$.
For any $a,b\in K$, we have $a^n=b^n=1$. Hence, $(ab^{-1})^n=a^n(b^{-1})^n=1,$
which implies $ab^{-1}\in K$ and hence $K$ is a subgroup of $\mathbb{F}_{q^k}^{\ast}$.

Now let $y\in \operatorname{Im}(f)$ and choose $x_0\in\mathbb{F}_{q^k}^{\ast}$ such that
$x_0^n=y$.
If $x\in f^{-1}(y)$, then
$x^n=x_0^n,$
and therefore
$(xx_0^{-1})^n=1.$
Thus,
$xx_0^{-1}\in K$,
which implies that there exists $k\in K$ such that
$x=x_0k$.
Conversely, for any $k\in K$,
$(x_0k)^n=x_0^nk^n=y.$
Hence,
$f^{-1}(y)=x_0K$.
Therefore, each nonempty preimage set of $f$ is a multiplicative coset of $K$.
\end{proof}
\begin{corollary}
Let $f \colon \mathbb{F}_{q^k} \to \operatorname{Im}(f)$ be defined by $f(x)=x^n$, and let $\ell=\gcd(n,q^k-1)$ and $M=\frac{q^k-1}{\ell}$. If $\beta$ is a primitive element of $\mathbb{F}_{q^k}$, then the preimage sets of $f$ on $\mathbb{F}_{q^k}^{\ast}$ are given by $\mathcal{A}_{\beta^g} = \{\beta^{g+pM} : 0 \le p \le \ell-1\},  0 \le g \le M-1.$ Equivalently, $\mathcal{A}_{\beta^g} = \left\{ \beta^g, \beta^{g+M}, \beta^{g+2M}, \dots, \beta^{g+(\ell-1)M} \right\}.$
Moreover, $\mathcal{A}_0=\{0\}.$ 
\end{corollary}

\begin{example}
\label{exm2}
Consider the monomial mapping $f(x)=x^n$ over $\mathbb{F}_{4^2}$, where $|\mathbb{F}_{4^2}^{*}|=15$ and $\ell=\gcd(n,15)$. Let $\beta$ be a primitive element of $\mathbb{F}_{4^2}$, and represent each field element as $(a,b)\in\mathbb{F}_4^2$.

\textbf{Case 1:} For $n=3$, $\ell=3$, so $f(x)=x^3$ is $3$-to-$1$. The $15$ nonzero elements form $5$ preimage sets of size $3$, yielding $|\mathrm{Im}(f)|=1+\frac{15}{3}=6$ with $\mathrm{Im}(f)=\{0,1,\beta^3,\beta^6,\beta^9,\beta^{12}\}$. The preimage sets are
\begin{align*}
\mathcal{A}_{0}
&= \{(0,0)\} \mapsto 0,\\
\mathcal{A}_{1}
&= \{(1,0), (\alpha,0), (\alpha^2,0)\} 
\leftrightarrow \{1, \beta^5, \beta^{10}\} \mapsto 1, \\[3pt]
\mathcal{A}_{\beta}
&= \{(0,1), (0,\alpha), (0,\alpha^2)\} \leftrightarrow \{\beta, \beta^6, \beta^{11}\} \mapsto \beta^3, \\[3pt]
\mathcal{A}_{\beta^2}
&= \{(\alpha,1), (\alpha^2,\alpha), (1,\alpha^2)\} \leftrightarrow \{\beta^2, \beta^7, \beta^{12}\} \mapsto \beta^6, \\[3pt]
\mathcal{A}_{\beta^3}
&= \{(\alpha,\alpha^2), (\alpha^2,1), (1,\alpha)\}  \leftrightarrow \{\beta^3, \beta^8, \beta^{13}\} \mapsto \beta^9, \\[3pt]
\mathcal{A}_{\beta^4}
&= \{(1,1), (\alpha,\alpha), (\alpha^2,\alpha^2)\}  \leftrightarrow \{\beta^4, \beta^9, \beta^{14}\} \mapsto \beta^{12}.
\end{align*}

\textbf{Case 2:} For $n=5$, $\ell=5$, so $f(x)=x^5$ is $5$-to-$1$. The $15$ nonzero elements form $3$ preimage sets of size $5$, yielding $|\mathrm{Im}(f)|=1+\frac{15}{5}=4$ with $\mathrm{Im}(f)=\{0,1,\beta^5,\beta^{10}\}$. The  preimage sets are
\begin{align*}
\mathcal{A}_{0}
&= \{(0,0)\} \mapsto 0,\\
\mathcal{A}_{1}
&= \{(1,0), (\alpha,\alpha^2), (0,\alpha), (\alpha,\alpha), (1,\alpha^2)\} \\
&\quad \leftrightarrow \{1, \beta^3, \beta^6, \beta^9, \beta^{12}\} \mapsto 1, \\[4pt]
\mathcal{A}_{\beta}
&= \{(0,1), (1,1), (\alpha^2,\alpha), (\alpha^2,0), (1,\alpha)\} \\
&\quad \leftrightarrow \{\beta, \beta^4, \beta^7, \beta^{10}, \beta^{13}\} \mapsto \beta^5, \\[4pt]
\mathcal{A}_{\beta^{2}}
&= \{(\alpha,1), (\alpha,0), (\alpha^2,1), (0,\alpha^2), (\alpha^2,\alpha^2)\} \\
&\quad \leftrightarrow \{\beta^2, \beta^5, \beta^8, \beta^{11}, \beta^{14}\} \mapsto \beta^{10}.
\end{align*}
\end{example}


%

\begin{example}
\label{exm3}
Consider the monomial mapping $f(x)=x^n$ over $\mathbb{F}_{3^2}$, where $|\mathbb{F}_{3^2}^{*}|=8$ and $\ell=\gcd(n,8)$. Let $\beta$ be a primitive element of $\mathbb{F}_{3^2}$, represent each field element as $(a,b)\in\mathbb{F}_3^2$.
Both cases yield the same image set and preimage partition; only the image assignments to the preimage sets differ.

\textbf{Case 1:} For $f(x)=x^2$, $\ell=\gcd(2,8)=2$. The preimage sets are
\begin{align*} \mathcal{A}_{0} &= \{(0,0)\} \mapsto 0,\\ \mathcal{A}_{1} &= \{(1,0), (2,0)\} \leftrightarrow \{1,\beta^4\} \mapsto 1, \\ \mathcal{A}_{\beta} &= \{(0,1), (0,2)\} \leftrightarrow \{\beta,\beta^5\} \mapsto \beta^{2}, \\ \mathcal{A}_{\beta^{2}} &= \{(1,2), (2,1)\} \leftrightarrow \{\beta^2,\beta^6\} \mapsto \beta^{4}, \\ \mathcal{A}_{\beta^{3}} &= \{(2,2), (1,1)\} \leftrightarrow \{\beta^3,\beta^7\} \mapsto \beta^{6}. \end{align*}

\textbf{Case 2:} For $f(x)=x^6$, $\ell=\gcd(6,8)=2$. The preimage sets are
\begin{align*} \mathcal{A}_{0} &= \{(0,0)\} \mapsto 0,\\ \mathcal{A}_{1} &= \{(1,0), (2,0)\} \leftrightarrow \{1,\beta^4\} \mapsto 1, \\ \mathcal{A}_{\beta} &= \{(2,2), (1,1)\} \leftrightarrow \{\beta^3,\beta^7\} \mapsto \beta^{2}, \\ \mathcal{A}_{\beta^{2}} &= \{(1,2), (2,1)\} \leftrightarrow \{\beta^2,\beta^6\} \mapsto \beta^{4}, \\ \mathcal{A}_{\beta^{3}} &= \{(0,1), (0,2)\} \leftrightarrow \{\beta,\beta^5\} \mapsto \beta^{6}. \end{align*}

\end{example}

\subsection{Plotkin-type Bound for monomial mapping function}
In this subsection, we derive the Plotkin-type bound for function-correcting codes for monomial mappings.

\begin{theorem}
        Let $f(x)=x^n$ be a monomial function over $\mathbb{F}_{q^k}$, and let $\ell=\gcd(n,q^k-1)$. The zero preimage set is the singleton $\mathcal A_0=\{0\}$, and every nonzero preimage set has cardinality $\ell$. Furthermore, if $\beta$ is a primitive element of $\mathbb F_{q^k}$, the distinct nonzero preimage sets are defined in Corollary~1. Let $M=\frac{q^k-1}{\ell}$ be the number of nonzero preimage sets. Let the kernel of the mapping be $K = \ker(f) = \{\omega \in \mathbb{F}_{q^k}^{\ast}: \omega^n = 1\}.$ We have  $\ell = |K| = \gcd(n, q^k-1).$ The redundancy satisfies the lower bound
        \begin{equation*}
        r_f(k,t) \ge \frac{(2t+1)(q^k-1)(q^k-\ell+1)+S}{q^{2k-1}(q-1)} - k,
        \end{equation*}
        where $S$ is the total intra-preimage Hamming distance given by
          $S=(\ell-1)k(q-1)q^{k-1}$.

\end{theorem}

\begin{proof}
               Since $\mathcal A_0=\{0\}$ has cardinality $1$, and each nonzero preimage set $\mathcal A_{\beta^g}$ for $0\le g\le M-1$ has cardinality $\ell$, we obtain $\sum_{g=1}^{m}\ell_g^2 = 1+\sum_{g=0}^{M-1}\ell^2 = 1+M\ell^2$. Substituting $\sum_{g=1}^{m}\ell_g^2=\ell(q^k-1)+1$ into the general bound of Theorem \ref{thm2} yields
         \begin{align}
                r_f(k,t)&\ge\frac{(2t+1)\bigl(q^{2k}-\ell(q^k-1)-1\bigr)+S}{q^{2k-1}(q-1)}-k\notag\\
                &=\frac{(2t+1)(q^k-1)(q^k-\ell+1)+S}{q^{2k-1}(q-1)}-k.
         \label{r_mono}
         \end{align}


Now we proceed to compute an expression for $S$. Since $\mathcal{A}_0 = \{0\}$ contains only one element, its intra-set distance is zero, allowing us to restrict the total distance sum exclusively to the nonzero classes from $g=0$ to $M-1$. Thus, $S=\sum_{g=0}^{M-1}\sum_{x,y\in A_{\beta^g}} d_H(x,y)$. Using the identity $d_H(u,v)=w_H(u-v)$ and writing each element of $A_{\beta^g}$ as $\beta^g\omega$, we obtain $S=\sum_{\omega\in K}\sum_{a=0}^{q^k-2}w_H\bigl(\beta^a(\omega-1)\bigr)$. 

Since the term corresponding to $\omega=1$ contributes zero, it simplifies to $S=\sum_{\substack{\omega\in K\\ \omega\neq1}}\sum_{a=0}^{q^k-2} w_H\bigl(\beta^a(\omega-1)\bigr)$. For any fixed $\omega\neq1$, multiplication by $\omega-1$ is a bijection on $\mathbb F_{q^k}^{*}$, and hence $\sum_{a=0}^{q^k-2}w_H\bigl(\beta^a(\omega-1)\bigr)=\sum_{z\in\mathbb F_{q^k}^{*}} w_H(z)$. Because $|K|=\ell$ and the total weight of the non-zero field elements evaluates to $\sum_{z\in\mathbb F_{q^k}^{*}} w_H(z)=k(q-1)q^{k-1}$, it follows directly that $S=(\ell-1)k(q-1)q^{k-1}$.

This completes the proof.
\end{proof}

\begin{example}

For the monomial mapping $f(x)=x^5$ over $\mathbb{F}_{2^4}$, we have
$\ell=\gcd(5,16-1)=\gcd(5,15)=5=n$. Hence, the $q^k-1=15$ nonzero elements are partitioned into $M=\frac{q^k-1}{\ell}=3$ preimage sets, each of size $\ell=5$. The term $S = \sum_{g=0}^{M-1} \sum_{x_i,x_j\in\mathcal A_{\beta^g}} d_H(x_i,x_j)$ represents the total Hamming distance within these preimage sets. One of the preimage sets is given by
$\mathcal{A}_{1}=\{1000, 0001, 0011, 0101, 1111\},$. The total Hamming distance obtained by summing, for each vector, its distance to all other vectors in the set is $44$. Repeating this computation for all three disjoint preimage sets yields $S=128$. Substituting this value into the redundancy bound in \eqref{r_mono} for t=1, we obtain $r\ge2$.

\end{example}


\section{Conclusion}
In this paper, we derived a generalised Plotkin-type redundancy bound under the Hamming metric for arbitrary functions. The proposed bound captures the influence of the underlying function through the structure of its preimage sets and recovers the classical Plotkin bound for error-correcting codes as a special case. We further obtained explicit bounds for monomial mapping functions of the form ($f(x)=x^n$) by characterising their preimage structure. We derived specialised bounds for linear, Hamming-weight, and Hamming-weight threshold functions. Several examples were presented to illustrate the applicability of the proposed bounds and the impact of function structure on the redundancy of function-correcting codes.

\end{document}